\documentclass[11pt]{article}
\usepackage[hmargin=3.5cm,vmargin=3cm]{geometry}
\usepackage{amsmath}
\usepackage{amssymb}
\usepackage{amsthm}
\usepackage[mathscr]{eucal}
\usepackage{enumitem}
\usepackage{graphicx}
\usepackage[pdftex]{color}
\definecolor{darkred}{rgb}{0.8,0.2,0.2}
\definecolor{darkblue}{rgb}{0.3,0.3,0.7}
\usepackage[pdftex,pdftitle={Arbitrage without borrowing or short selling?},pdfauthor={Jani Lukkarinen and Mikko S. Pakkanen},colorlinks=TRUE,allcolors=darkred]{hyperref}
\usepackage[text]{datetime}
\newdateformat{ukvardate}{%
\THEDAY{} \monthname[\THEMONTH] \THEYEAR}

\numberwithin{equation}{section}

\newtheorem{theorem}[equation]{Theorem}
\newtheorem{lemma}[equation]{Lemma}
\newtheorem{prop}[equation]{Proposition}

\theoremstyle{definition}

\newtheorem{open}[equation]{Open Question}

\theoremstyle{remark}
\newtheorem{remark}[equation]{Remark}

\newcommand{\ud}{\mathrm{d}}

\newcommand{\eqdefl}{\mathrel{\mathop:}=}
\newcommand{\N}{\mathbb{N}}

\newcommand{\R}{\mathbb{R}}
\newcommand{\E}{\mathbb{E}}

\newcommand{\prob}{\mathbb{P}}

\begin{document}
\title{Arbitrage without borrowing or short selling?}
\date{\ukvardate\today}

\author{Jani Lukkarinen\thanks{Department of Mathematics and Statistics, University of Helsinki, P.O.\ Box 68, FI-00014 Helsingin yliopisto, Finland.
E-mail: \href{mailto:jani.lukkarinen@helsinki.fi}{\nolinkurl{jani.lukkarinen@helsinki.fi}}} \and
Mikko S. Pakkanen\thanks{
Department of Mathematics, 
Imperial College London, 
South Kensington Campus,
London SW7 2AZ, UK and CREATES, Aarhus University, Denmark.
E-mail:\ 
\href{mailto:m.pakkanen@imperial.ac.uk}{\nolinkurl{m.pakkanen@imperial.ac.uk}}}
}

\maketitle

\begin{abstract}

We show that a trader, who starts with no initial wealth and is not allowed to borrow money or short sell assets, is theoretically able to attain positive wealth by continuous trading, provided that she has perfect foresight of future asset prices, given by a continuous semimartingale. Such an arbitrage strategy can be constructed as a  process of finite variation that satisfies a seemingly innocuous self-financing condition, formulated using a pathwise Riemann--Stieltjes integral. Our result exemplifies the potential intricacies of formulating economically meaningful self-financing conditions in continuous time, when one leaves the conventional arbitrage-free framework.

\vspace*{1em}


\vspace*{1em}


\vspace*{1em}


\end{abstract}

\section{Introduction}

Common sense suggest that arbitrage strategies --- in the sense of mathematical finance, involving no initial wealth --- should require short selling or an access to credit --- an obvious \emph{budget constraint}. Indeed, in the real world, and in discrete-time models as well, we can distinguish the first position in the risky asset prescribed by the strategy. If this position were not short, it would have to be funded by borrowed money. However, in the realm of continuous trading, there might not be any ``first position'', as the composition of the portfolio can vary rather freely as a function of time, so it is not a priori clear if arbitrage strategies without short selling or borrowing are impossible.

\emph{Self-financing conditions} are an important aspect of dynamic trading strategies. They should be seen as a means to enforce coherent accounting: All profits from trading must be credited to, and all trading costs debited from the money market account. In continuous time, self-financing conditions are formulated using stochastic integrals; see, e.g., Bj\"ork \cite[Sections 6.1 and 6.2]{Bjo2009}. In particular, for adapted strategies, \emph{It\^o integrals} can be used when the price process is a semimartingale. However, the choice of the integral is a rather delicate matter, as not all stochastic integrals lend themselves to economically meaningful self-financing conditions. (For example, the paper by Bj\"ork and Hult \cite{BH2005} documents some interpretability issues that arise from the use of \emph{Skorohod integrals} and \emph{Wick products} in self-financing conditions.) In any case, any sound self-financing condition should at the very least rule out arbitrage strategies without short selling or borrowing. After all, such trading strategies, which are able to generate wealth literally \emph{ex nihilo}, should definitely not be self-financing.

Besides It\^o integration, \emph{pathwise Riemann--Stieltjes integrals} (see, e.g., Riga \cite{Rig2016}, Salopek \cite{Sal1998}, or Sottinen and Valkeila \cite{SV2003}) have often been seen as a ``safe'' way to formulate reasonable self-financing conditions. The reasons are manifold: Like It\^o integrals, Riemann--Stieltjes integrals can, of course, be obtained transparently as limits of Riemann sums that reflect the natural self-financing condition for simple trading strategies. Also, a pathwise Riemann--Stieltjes integral coincides with the corresponding It\^o integral whenever the latter exists. Recall that a Riemann--Stieltjes integral is guaranteed to exist for example when the integrator is continuous and the integrand is of \emph{finite variation}. While it typically rules out the Markovian trading strategies that arise in dynamic hedging and utility maximisation, say, the finite variation assumption is economically justified as it amounts to keeping the trading volume of the strategy finite (which is an essential requirement under transaction costs); see, e.g., Longstaff \cite{Lon2001}.

However, it transpires that self-financing conditions based on pathwise Riemann--Stieltjes integrals alone do not necessarily prohibit pathological trading strategies (even of finite variation). We show in this note that, quite surprisingly, a Riemann--Stieltjes-based self-financing condition may in fact admit arbitrage strategies that require neither borrowing nor short selling if the trader has \emph{perfect foresight} of the future prices of the risky asset.\footnote{In many cases, it is actually sufficient to have perfect foresight only on an arbitrarily short time interval, as is pointed out in Remark \ref{rem:mainres}.} Our existence result for such strategies (Theorem \ref{thm:main}, below) is valid provided that the price process is a continuous semimartingale with an equivalent local martingale measure and non-degenerate quadratic variation. While the requirement of perfect foresight is admittedly unusual, a sound self-financing condition should nevertheless prevent even a perfectly informed trader from executing such an egregious arbitrage strategy. More importantly, from a mathematical perspective, this result illustrates how stochastic integrals, even when defined pathwise, may not always behave as financial intuition would suggest. We additionally show that these arbitrage strategies would in fact not be possible if also the price process were of finite variation (Proposition \ref{prop:finite}, below). This indicates that the phenomenon documented in this note is intricately linked with the fine properties and ``roughness'' of the price process. 


\section{Model and main results}

Let us consider a continuous-time market model with a risky asset and a risk-free money market account, where trading is possible up to a finite time horizon $T \in (0,\infty)$. The price of the risky asset follows a continuous, positive-valued semimartingale 	$S=(S_t)_{t \in [0,T]}$, defined on a complete probability space $(\Omega,\mathscr{F},\prob)$. For simplicity, the interest rate of the money market account is zero. Additionally, we denote by $(\mathscr{F}^S_t)_{t \in [0,T]}$ the natural filtration of the price process $S$, augmented the usual way to make it complete and right-continuous, and by $\langle S \rangle$ the quadratic variation process of $S$. Throughout the paper, we use the interpretation $\inf \varnothing = \infty$.

Consider a trader, whose trading strategy is described by two c\`agl\`ad (continuous from left with limits from right) processes $\psi = (\psi_t)_{t \in [0,T]}$ and $\phi = (\phi_t)_{t \in [0,T]}$ that keep track of her money market account balance and holdings in the risky asset, respectively. The mark-to-market value of her portfolio at time $t \in [0,T]$ can then be expressed as
\begin{equation}\label{eq:value}
V_t = \psi_t + \phi_t S_t\, .
\end{equation}
As per the discussion above, we are interested in a scenario where the trader attempts to follow an arbitrage strategy, so she starts with no initial wealth, which translates to the constraint $V_0=0$.

The trader is additionally subject to a self-fi\-nan\-cing condition. Let us assume provisionally that $\psi$ and $\phi$ are adapted to the filtration $(\mathscr{F}^S_t)_{t \in [0,T]}$. Then the self-financing condition is formulated the usual way \cite[Sections 6.1 and 6.2]{Bjo2009} by requiring that
\begin{equation}\label{eq:selfie}
V_t = V_0 + \int_0^t \phi_u \ud S_u =  \int_0^t \phi_u \ud S_u \quad \textrm{for any $t \in [0,T]$}\, ,
\end{equation}
where the integral with respect to $S$ is understood as an It\^o integral. Under the self-financing condition \eqref{eq:selfie}, the process $\psi$ becomes redundant as, by plugging \eqref{eq:selfie} into \eqref{eq:value}, we can solve for $\psi_t$, to wit,
\begin{equation}\label{eq:psi}
\psi_t = \int_0^t \phi_u \ud S_u - \phi_t S_t, \quad  t \in [0,T]\, .
\end{equation}

Now the key question is: Are there non-trivial processes $\phi$, with $\phi_t \geqslant 0$ for all $t \in [0,T]$, such that $\psi_t \geqslant 0$ for all $t \in [0,T]$? Using \eqref{eq:psi}, we can reformulate this as a question of existence of non-negative processes $\phi$ that satisfy the stochastic inequality
\begin{equation}\label{eq:inequality}
\int_0^t \phi_u \ud S_u \geqslant \phi_t S_t\quad \textrm{for all $t \in [0,T]$}\, .
\end{equation}
In the adapted case, we can answer the question in a straightforward manner if we assume that $S$ is arbitrage-free. Indeed, if there exists a probability measure $\mathbb{Q}$ on $(\Omega,\mathscr{F})$ such that $\mathbb{Q} \sim \prob$ (where ``$\sim$'' denotes mutual absolute continuity of measures, as usual) and that $S$ is a local $\mathbb{Q}$-martingale, then a suitable version of the \emph{fundamental theorem of asset pricing} (e.g., \cite[Corollary 1.2]{DS1994}) implies that there are no non-negative (adapted) processes $\phi$ that would satisfy \eqref{eq:inequality} and $\prob(V_t > 0)>0$ for some $t \in (0,T]$.

However, as discussed above, we shall not insist on adaptedness, so we consider processes $\phi$ that are not necessarily adapted to the filtration $(\mathscr{F}^S_t)_{t \in [0,T]}$. Then the stochastic integral with respect to $S$ that appears in \eqref{eq:selfie}, \eqref{eq:psi} and \eqref{eq:inequality} may not exist as an It\^o integral. But if we assume that $\phi$ is of finite variation, then the integral does exist as a Riemann--Stieltjes integral, see \cite[Theorems 1.2.3 and 1.2.13]{Str2011}, defined path-by-path for any $t \in [0,T]$ by
\begin{equation}\label{eq:RS-int}
\int_0^t \phi_u \ud S_u \eqdefl \lim_{n \rightarrow \infty} \sum_{i=1}^{k_n}\phi_{\tau^n_{i}} (S_{t \wedge \tau^n_i}-S_{t \wedge \tau^n_{i-1}}) \quad \textrm{$\prob$-a.s.}\, ,
\end{equation}
where $x \wedge y \eqdefl \min \{ x,y\}$ for all $x,y \in \R$ and $(\tau^n_i)_{i=0,n \geqslant 1}^{k_n}$ is a family of random times such that
\begin{equation*}
0 = \tau^n_0 \leqslant \tau^n_1 \leqslant \cdots \leqslant \tau^n_{k_n} = T, \quad \textrm{for any $n \geqslant 1$}\, ,
\end{equation*}
and $\lim_{n \rightarrow \infty} \sup_{1 \leqslant i \leqslant k_n} (\tau^n_i - \tau^n_{i-1})=0$. The definition \eqref{eq:RS-int} is independent of the choice of $(\tau^n_i)_{i=0,n \geqslant 1}^{k_n}$. Further, it ensures that the self-financing condition based on such integrals reduces to the usual self-financing condition when $\phi$ is \emph{simple}, that is, piecewise constant.

\begin{remark}
While trading strategies in mathematical finance literature are conventionally assumed to be adapted to the natural filtration of the price process, non-adapted strategies do appear in literature on insider trading; see, e.g., \cite{AIS1998,PK1996}. More recently, it has also been suggested that (imprecise) prior information of future price changes at very short time scales may be available to high-frequency traders and market makers \cite{Hir2016}.
\end{remark}

Our main result shows that, in this alternative framework, there are in fact non-trivial, non-negative processes $\phi$ that satisfy the inequality \eqref{eq:inequality}. The proof of this result is carried out in Section \ref{sec:proof}, below.

\begin{theorem}\label{thm:main}
Suppose that the positive continuous semimartingale $S=(S_t)_{t \in [0,T]}$ satisfies $\prob(\langle S \rangle_T > 0)>0$. Assume further that there exist a probability measure $\mathbb{Q} \sim \prob$ such that $S$ is a local $\mathbb{Q}$-martingale.
Then there exists a non-negative process $\phi=(\phi_t)_{t \in [0,T]}$, with c\`agl\`ad sample paths of finite variation, such that $\phi_0=0$, and
\begin{subequations}
\begin{align}
&\int_0^t \phi_u \ud S_u \geqslant \phi_t S_t \textrm{ for all $t \in [0,T]$ $\prob$-a.s.}\, , \label{eq:as}\\
&\prob\bigg( \int_0^t \phi_u \ud S_u > \phi_t S_t \textrm{ for all $t \in (\rho,T]$}\bigg)  > 0\, ,\label{eq:prob}
\end{align}
\end{subequations}
where $\rho \eqdefl \inf \{ t \in (0,T] : \langle S \rangle_t > 0\} \wedge T$.
\end{theorem}

\begin{figure}[!p] 
\centering
\begin{tabular}{rl} 
\includegraphics[scale=0.84,trim=0.6cm 1.5cm 0.7cm 1cm, clip=TRUE]{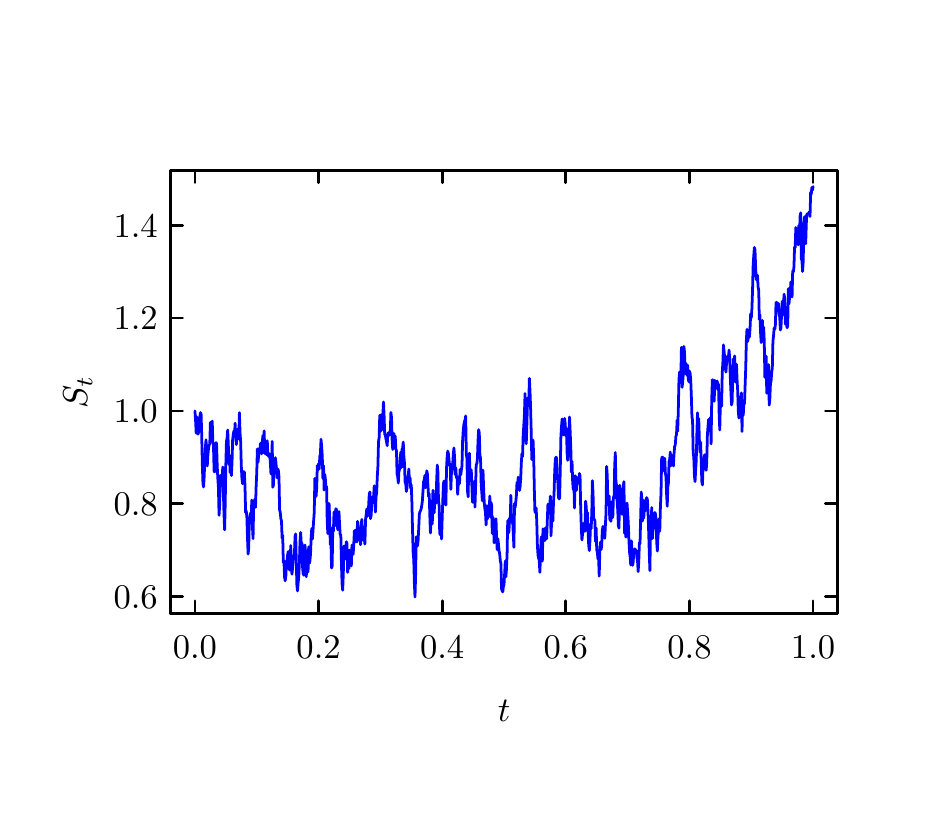} & \includegraphics[scale=0.84,trim=0.6cm 1.5cm 0.7cm 1cm, clip=TRUE]{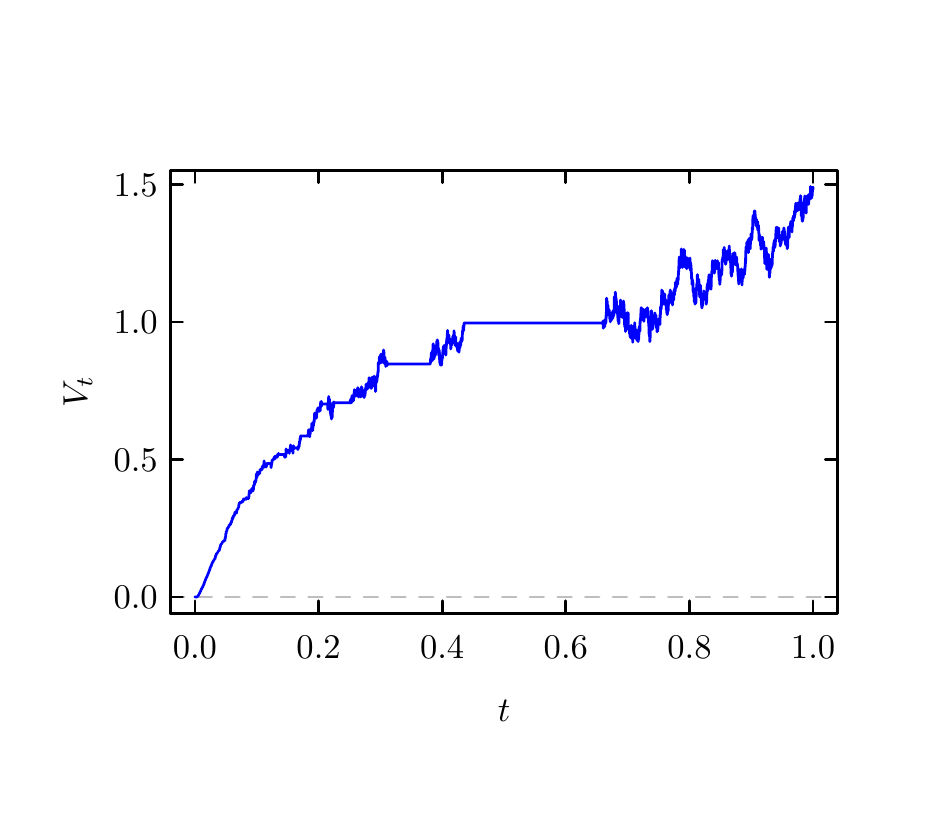} \\
\includegraphics[scale=0.84,trim=0.6cm 1cm 0.7cm 1.5cm, clip=TRUE]{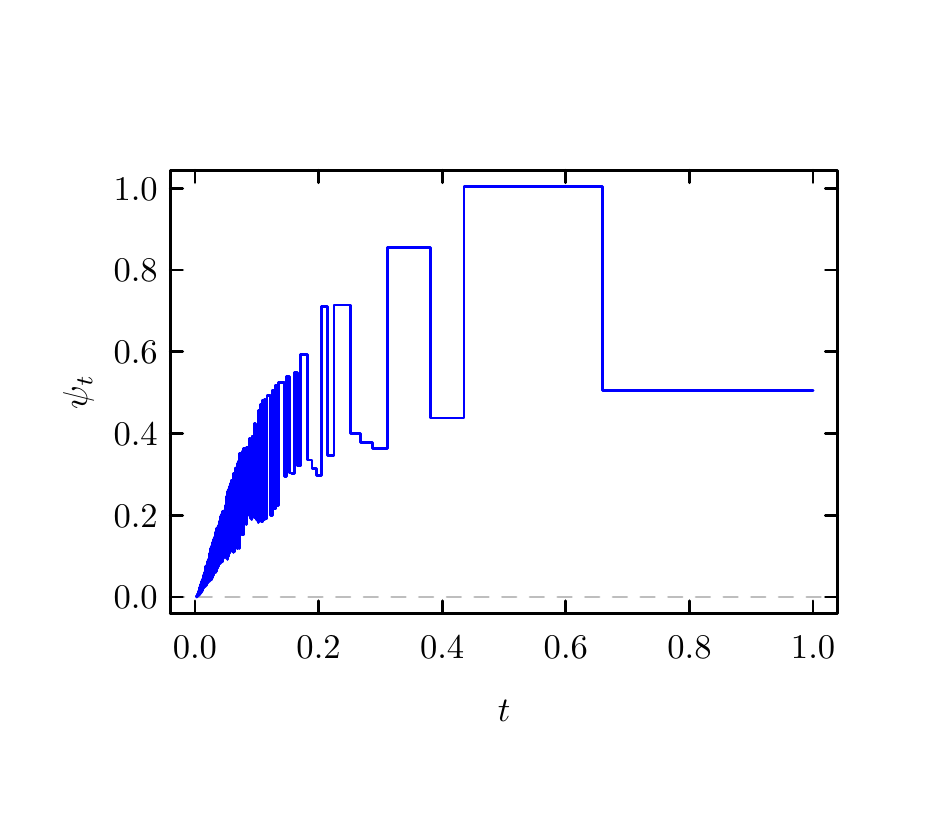} & \includegraphics[scale=0.84,trim=0.6cm 1cm 0.7cm 1.5cm, clip=TRUE]{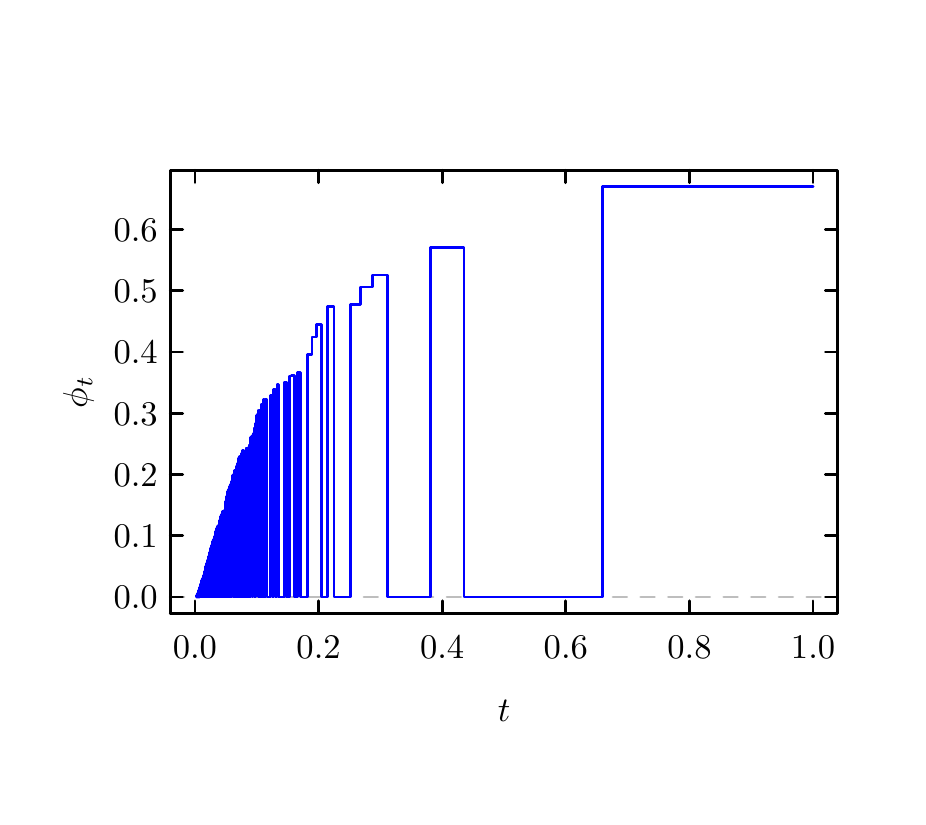}
\end{tabular}
\caption{Numerical illustration of Theorem \ref{thm:main}. In this example $T = 1$ and $S$ is a Brownian motion started at one, so that $\rho=0$. (Theoretically, the requirement that $S$ is positive can then be met, e.g., by reflecting or absorbing the process at some level between zero and one.)
The realisation of the process $\phi$ has been generated following the construction \eqref{K-def} given in the proof of Theorem \ref{thm:main}, below. Recall that $\psi_t = \int_0^t \phi_u \ud S_u - \phi_t S_t$ and $V_t = \psi_t + \phi_t S_t=\int_0^t \phi_u \ud S_u$.\label{fig:1}} 
\end{figure}

\begin{remark}\label{rem:mainres}
\begin{enumerate}[label=(\roman*),ref=\roman*,leftmargin=2.2em]
\item While not explicitly stated above, the process $\phi$ of Theorem \ref{thm:main} is indeed not (and could not be) adapted to $(\mathscr{F}^S_t)_{t \in [0,T]}$. The specification of $\phi_t$ for any $t\in(0,T]$ requires full knowledge of the path of $S$ until time $T$.
However, the process $\phi$ is adapted to the filtration
\begin{equation*}
\tilde{\mathscr{F}}^S_t \eqdefl \mathscr{F}^S_T, \quad t \in [0,T]\, ,
\end{equation*}
corresponding to perfect foresight on $S$, which also ensures that $\phi$ does not depend on any (external) randomness beyond $S$.
\item It is also worth stressing that the time horizon $T\in (0,\infty)$ can be chosen freely, as long as $\prob(\langle S \rangle_T > 0)>0$ is satisfied. In particular, if $S$ has strictly increasing quadratic variation, then we can choose $T$ to be arbitrarily small --- that is, prior knowledge of the fluctuations of $S$ is required only on a very short time interval. 
\item In mathematical finance literature, it is common to restrict trading strategies to be \emph{admissible}; see, e.g., \cite[Definition 2.7]{DS1994}. While there are actually several slightly differing definitions of admissibility, they have the commonality that the value process of an admissible strategy is bounded from below (in some sense). The purpose of admissibility conditions is to preclude some outright pathological trading strategies, such as \emph{doubling strategies} \cite[p.\ 467]{DS1994}. It is worth stressing that the process $\phi$ of Theorem \ref{thm:main} would \emph{not} violate the typical admissibility conditions as the corresponding value process $V_t = \int_0^t \phi_u \ud S_u$, $t \in [0,T]$, is non-negative due to the property \eqref{eq:as}. 
\end{enumerate}
\end{remark}

Curiously, the assumption about positive quadratic variation in Theorem \ref{thm:main} --- that is, $S$ exhibits ``enough'' fluctuation --- is rather crucial: Using a result \cite[Theorem 3.1]{LP2013} on the positivity of Riemann--Stieltjes integrals, we can show that arbitrage without borrowing or short selling is in fact eliminated in this setting if also the price process $S$ is of finite variation:

\begin{prop}\label{prop:finite}
Suppose that the positive continuous semimartingale $S=(S_t)_{t \in [0,T]}$ satisfies, $\prob$-a.s., $\langle S \rangle_T=0$.  If $\phi=(\phi_t)_{t \in [0,T]}$ is a non-negative process with c\`agl\`ad sample paths of finite variation such that
\begin{equation}\label{eq:phipositive}
\mathbb{P}( \phi_t >0 \textrm{ for some $t \in [0,T]$})>0\, ,
\end{equation}
then
\begin{equation}\label{eq:positiveleverage}
\prob\bigg(\int_0^t \phi_u \ud S_u < \phi_t S_t \textrm{ for some $t \in [0,T]$}\bigg)>0\, .
\end{equation}
\end{prop}

\begin{proof}
The integration by parts formula for Riemann--Stieltjes integrals \cite[Theorem 1.2.3]{Str2011} yields
\begin{equation*}
\int_0^t \phi_u \ud S_u = \phi_t S_t - \phi_0 S_0 - \int_0^t S_u \ud \phi_u, \quad t \in [0,T]\, .
\end{equation*}
Note that since $S$ is a continuous semimartingale, the assumption $\langle S \rangle_T=0$ implies that the sample paths of $S$ are of finite variation. Now if $\phi_t >0$ for some $t \in [0,T]$, then it follows\footnote{The term \emph{non-vanishing} in the statement of \cite[Theorem 3.1]{LP2013} is potentially misleading. The appropriate interpretation is that \emph{the integrand $g$ should not be identically zero}. It is also worth mentioning that the assumption $g(a)=0$ therein can be trivially weakened to $g(a)\geqslant 0$; see \cite[p.\ 401]{LP2013}.} from \cite[Theorem 3.1]{LP2013} that
\begin{equation*}
\int_0^t S_u \ud \phi_u > 0 \quad \textrm{for some $t \in [0,T]$}\, .
\end{equation*}
The probability \eqref{eq:positiveleverage} is thus greater than or equal to the probability \eqref{eq:phipositive}, and the assertion follows.
\end{proof}

\begin{remark}
In some way, Theorem \ref{thm:main} and Proposition \ref{prop:finite} defy the usual mathematical finance intuition that ``smooth'' price processes are easier to arbitrage than ``rough'' ones. Here the ``roughness'' of $S$ is the very property that makes it possible to construct the process $\phi$ in Theorem \ref{thm:main}.
\end{remark}

In Theorem \ref{thm:main}, we assume that the price process $S$ is arbitrage-free whilst the strategy $\phi$ may not be adapted. This is, of course, only one of the possible departures from the standard arbitrage-free setting. Alternatively, one could also consider a scenario where the process $S$ is a very general continuous process that may admit arbitrage and $\phi$ is an adapted strategy of finite variation and ask, how the stronger form of arbitrage without short selling and borrowing can be excluded. This looks less straightforward and may require some new techniques and estimates for Riemann--Stieltjes integrals, so we leave the question open:

\begin{open}
When $S$ is a general positive, continuous process (not necessarily a semimartingale), under which conditions on $S$ is arbitrage without borrowing or short selling excluded in the context of strategies of finite variation? We remark that, to this end, the process $S$ should satisfy some kind of a non-degeneracy condition, as integrands similar to $\phi$ of Theorem \ref{thm:main} can be constructed for deterministic continuous paths that exhibit enough variation; see \cite[Theorem 2.1]{LP2013}.

\end{open}

\section{Proof of Theorem \ref{thm:main}}\label{sec:proof}

Before proving Theorem \ref{thm:main} rigorously, we describe intuitively how the process $\phi$ is constructed. The idea is to structure $\phi$ from a sequence of non-overlapping static positions in the risky asset, so that they have an ``accumulation point'' at $\rho$, see Figure \ref{fig:1}, bottom-right panel, for an illustration. The sizes of these static positions are chosen so that they are gradually increasing (from zero) and the positions are timed, using the quadratic variation of $S$ and perfect foresight, so that the price of the asset is known to increase during each holding period.

While the construction of $\phi$ this way is simple in principle, it is non-trivial to select the sizes of the static positions so that:
\begin{itemize}
\item each position can be fully funded using the profits from the preceding positions (without needing to borrow money),
\item the cumulative trading volume remains finite, which is equivalent to $\phi$ being of finite variation.
\end{itemize}
In fact most of the theoretical arguments in the proof of Theorem \ref{thm:main} revolve around verifying that these two requirements are indeed met.

We introduce now some additional notation that are needed in the sequel.
For all $x,y \in \R$, we denote $x \vee y \eqdefl \max \{ x,y\}$ and $x^+ \eqdefl x \vee 0$. 
If $X$ and $Y$ are identically distributed random variables, we write $X\stackrel{d}{=}Y$.
Suppose that $A \in \mathscr{F}$. Then we say that a property $\mathscr{P}$ (provided that it is ``$\mathscr{F}$-measurable'') holds $\prob$-a.s.\ on $A$, if $\prob(\{\mathscr{P} \} \cap A) = \prob(A)$. 
We use the convention that $\N \eqdefl \{1,2,\ldots\}$.

As a preparation, we prove now two technical lemmata, which will be instrumental in the proof of Theorem \ref{thm:main}.

\begin{lemma}\label{lem:verification}
Let $(y_n)_{n=1}^\infty$ be a sequence of non-negative numbers such that $\lim_{n\rightarrow \infty} y_n = 0$. Suppose 
that for some $\alpha \in (0,1)$,
\begin{equation}\label{eq:ysumassump}
\sum_{n=1}^\infty e^{-\alpha \sum_{k=1}^n\limits y_k}<\infty\, .
\end{equation}
If we set $\beta \eqdefl \frac{2 \alpha}{1+\alpha}$ and define a sequence $(x_n)_{n=1}^\infty$ of non-negative numbers by
\begin{equation}\label{eq:defxnseq}
x_n \eqdefl \prod_{k=1}^n \frac{1}{1+\beta y_k}, \quad n\in \N\, ,
\end{equation}
then $\sum_{n=1}^\infty x_n <\infty$ and 
\begin{equation}\label{eq:xnineq}
x_n < \sum_{k = n+1}^\infty x_k y_k <\infty \quad \textrm{for any $n \in \N$}\, .
\end{equation}
\end{lemma}

\begin{proof}
Consider a sequence $(y_n)_{n=1}^\infty$ and $\alpha\in (0,1)$ that satisfy the assumptions given above.  Define
then $\beta \eqdefl \frac{2 \alpha}{1+\alpha}$ and a sequence $(x_n)_{n=1}^\infty$ through \eqref{eq:defxnseq}.  Clearly, then
$0<\alpha<\beta<1$ and $0<x_n\leqslant 1$ for all $n \in \N$.

By the definition \eqref{eq:defxnseq},
\begin{equation*} 
x_{n-1}=(1+\beta y_n) x_n = x_n + \beta x_n y_n \quad \textrm{for $n\geqslant 2$}\, .  
\end{equation*}
Thus for all $n,N\in \N$ with $n<N$, we have
\begin{align}\label{eq:xysump}
 \beta \sum_{k=n+1}^N x_k y_k = x_n-x_N \, .
\end{align}
If $\sum_{k=1}^\infty x_k <\infty$, then necessarily $x_k\to 0$ as $k\to \infty$.  Therefore, we 
may take $N\to \infty$ in (\ref{eq:xysump}) and conclude that for every $n\in \N$,
\begin{align*}
 x_n <  \beta^{-1} x_n = \sum_{k=n+1}^\infty x_k y_k<\infty \, ,
\end{align*}
and, thus, (\ref{eq:xnineq}) holds then.
To complete the proof,
it remains to show that $\sum_{k=1}^\infty x_k <\infty$.

To prove that the assumptions on $(y_n)_{n=1}^\infty$ imply the summability of the sequence $(x_n)_{n=1}^\infty$, 
we rely on the inequality
$\log (1+x) \geqslant x/(1+x)$, which holds for all $x \geqslant 0$. (This inequality can be proven using the integral representation $\log(1+x)=\int_0^x (1+y)^{-1}\ud y$ 
and the monotonicity of the integrand therein.)  Hence, for all $k\geqslant n\geqslant 1$ we have
\begin{align*}
 -\log (1+\beta y_k) \leqslant -\gamma_n \beta y_k =- \frac{2 \gamma_n}{1+\alpha} \alpha y_k \, ,
\end{align*}
where $\gamma_n := 1/(1+\beta  \sup_{m\geqslant n} y_m)$.  Since $\lim_{k\to \infty} y_k= 0$, here $\gamma_n\nearrow 1$ as $n\to\infty$.
In particular, there exists $n_0\in \N$ such that $\gamma_{n_0} \geqslant (1+\alpha)/2$ and for all $k\geqslant n_0$ we then have 
$-\log(1+\beta y_k) \leqslant -\alpha y_k$.
Using the above estimates to the representation $x_n = \prod_{k=1}^n e^{-\log(1+\beta y_k)}$ 
thus shows that for all $n> n_0$, 
\begin{align*}
 0< x_n \leqslant x_{n_0} \exp\Bigl(-\alpha \sum_{k=n_0+1}^n y_k\Bigr) =
 x_{n_0} \exp\Bigl(\alpha \sum_{k=1}^{n_0} y_k\Bigr)  \exp\Bigl(-\alpha \sum_{k=1}^n y_k\Bigr) 
 \, .
\end{align*}
So (\ref{eq:ysumassump}) ensures that, indeed, $\sum_{n=1}^\infty x_n <\infty$.
\end{proof}

\begin{lemma}\label{lem:BM-conditions}
Let $B=(B_t)_{t \in [0,1]}$ be a standard Brownian motion and let $\sigma > 0$. If we define for some $\gamma \in (0,1)$, 
\begin{equation*}
\xi_n \eqdefl \big(\sigma B_{n^{-\gamma}}-\sigma B_{(n+1)^{-\gamma}}\big)^+, \quad n \in \N\, ,
\end{equation*}
then for any $\alpha >0$,
\begin{equation*}
\E\Bigg( \sum_{n=1}^\infty e^{-\alpha \sum\limits_{k=1}^n \xi_k}\Bigg)<\infty\, .
\end{equation*}
\end{lemma}

\begin{proof} 
By the self-similarity of Brownian motion, we have $\xi_n \stackrel{d}{=} u_n B^+_1$ for any $n \in \N$, where
\begin{equation*}
u_n \eqdefl \sigma \sqrt{n^{-\gamma} - (n+1)^{-\gamma}}\, .
\end{equation*}
Applying the mean value theorem to the function $x \mapsto x^{-\gamma}$, we deduce that 
\begin{equation}\label{eq:mvtbounds}
\gamma^{\frac{1}{2}} \sigma (n+1)^{-p} \leqslant u_n  \leqslant \gamma^{\frac{1}{2}}\sigma n^{-p}\, ,
\end{equation}
with $p \eqdefl \frac{\gamma+1}{2} \in \big(\frac{1}{2},1\big)$.
 
Let us now fix $\alpha>0$. By Tonelli's theorem and the mutual independence of the random variables $\xi_1,\xi_2,\ldots$, we obtain
\begin{equation}\label{eq:phiproduct}
\E\Bigg(\sum\limits_{n=1}^\infty e^{-\alpha \sum\limits_{k=1}^n \xi_k}\Bigg) = \sum\limits_{n=1}^\infty\E\bigg(e^{-\alpha \sum\limits_{k=1}^n \xi_k} \bigg) = \sum\limits_{n=1}^\infty \prod_{k=1}^n \varphi(-\alpha u_n)\, ,
\end{equation} 
where $\varphi(u) \eqdefl \E\big(e^{u B^+_1}\big)$, $u \in \R$. Since the  $\prob(B^+_1\geqslant 0)=1$, $\E(B^+_1)<\infty$
and $\prob(B^+_1>0) = \frac{1}{2}>0$, we have
\begin{align*}
 c_1 \eqdefl \inf_{v\in (0, \alpha\sigma)}
 \frac{\E\big(B^+_1 e^{-v B^+_1}\big)}{\E\big(e^{-v B^+_1}\big)}
\in (0,\infty)
 \,. 
\end{align*}
By Jensen's inequality, for $u\in (0,\alpha\sigma)$,
\begin{align*}
 \frac{1}{\varphi(-u)} = \frac{\E\big(e^{u B^+_1} e^{-u B^+_1}\big)}{\E\big(e^{-u B^+_1}\big)} \geqslant
\exp\Bigg(u\frac{\E\big(B^+_1 e^{-u B^+_1}\big)}{\E\big(e^{-u B^+_1}\big)}\Bigg)
 \geqslant
\exp(u c_1)\, ,
\end{align*}
which implies that 
\begin{equation}\label{eq:prod-exp}
 \prod_{k=1}^n \varphi(-\alpha u_k) \leqslant \exp\bigg(-c_1  \sum_{k=1}^n u_k\bigg)\, ,
\end{equation}
for any $n \in \N$.

Using the lower bound in \eqref{eq:mvtbounds} we can estimate, for any $n \in \N$, 
\begin{equation}\label{eq:u-lower}
 \sum_{k=1}^n u_k \geqslant \gamma^{\frac{1}{2}} \sigma \sum_{k=1}^n (k+1)^{-p} \geqslant \gamma^{\frac{1}{2}} \sigma \int_2^{n+2} x^{-p} \ud x \geqslant c_2 n^{1-p}\, ,
\end{equation}
where $c_2 = c_2(\gamma,\sigma,p)>0$ is a constant. Now note that for any exponent $\theta>0$, there exists a constant $c_3 = c_3(\theta)>0$ such that $e^x \geqslant c_3 x^\theta$, $x\geqslant 0$. Thus, applying \eqref{eq:u-lower} to \eqref{eq:prod-exp}, we find that for any $n \in \N$,
\begin{equation*}
\prod_{k=1}^n \varphi(-\alpha u_k) \leqslant e^{-c_1 c_2  n^{1-p}} \leqslant \frac{1}{(c_1 c_2 )^{\theta}c_3} \frac{1}{n^{\theta(1-p)}}\, ,
\end{equation*}
and, in view of \eqref{eq:phiproduct}, it remains to choose $\theta > \frac{1}{1-p}$.
\end{proof}

The proof of Theorem \ref{thm:main} is based on the observation that the properties \eqref{eq:as} and \eqref{eq:prob} the process $\phi$ is expected to satisfy are robust to time changes and equivalent changes of the probability measure. Under the assumptions of Theorem \ref{thm:main}, we can represent the process $S$ as a time-changed Brownian motion under an equivalent local martingale measure. Therefore we can verify \eqref{eq:as} and \eqref{eq:prob} relying on the properties of Brownian motion via Lemmata \ref{lem:verification} and \ref{lem:BM-conditions}.

\begin{proof}[Proof of Theorem \ref{thm:main}]
The properties of the process $\phi$ to be constructed are clearly invariant under rescaling of the process $S$ by a positive constant. By rescaling $S$, the probability of the event $\{\sup_{t \in [0,T]} S_t \leqslant 1\}$ can be made to be arbitrarily close to one. In particular, we may assume, without loss of generality, that
\begin{equation*}
\prob\Big(\sup_{t \in [0,T]} S_t \leqslant 1\Big) > 1- \prob(\langle S \rangle_T >0)\, ,
\end{equation*}
where, by assumption, $\prob(\langle S \rangle_T >0)$ is positive and remains so even after the process $S$ has been rescaled.
This implies that $\prob(\sup_{t \in [0,T]} S_t \leqslant 1,\, \langle S \rangle_T > 0)>0$, so we can find a constant $c>0$ such that the event $A_c \eqdefl \{\sup_{t \in [0,T]} S_t \leqslant 1\} \cap \{ \langle S \rangle_T > c\}$ satisfies $\prob(A_c)>0$.

Let now $\gamma \in (0,1)$, and introduce the stopping times
\begin{equation*}
\rho_n \eqdefl \inf \{ t \in [0,T] : \langle S \rangle_t \geqslant cn^{-\gamma} \} \wedge T, \quad n \in \N\, .
\end{equation*}
Since the process $S$ is continuous, also its quadratic variation process $\langle S \rangle$ is $\prob$-a.s.\ continuous \cite[Theorem 17.5]{Kal2002}. Thus we have $\prob$-a.s.\ on $A_c$,
\begin{equation*}
0 \leqslant \rho < \cdots < \rho_n<\cdots < \rho_2 < \rho_1 < T
\end{equation*}
and $\rho_n \searrow \rho$ as $n\rightarrow \infty$. Note additionally that
\begin{equation}\label{rho-times}
\langle S \rangle_{\rho_n} = cn^{-\gamma}\quad \textrm{$\prob$-a.s.\ on $A_c$ for any $n \in \N$}\, .
\end{equation}
We also introduce the random variables
\begin{equation*}
Z_n \eqdefl (S_{\rho_{n}} - S_{\rho_{n+1}})^+,\quad n \in \N\, ,
\end{equation*}
which will be instrumental in what follows.

Let now $\mathbb{Q} \sim \prob$ be such that $S$ is a local $\mathbb{Q}$-martingale. Then, clearly, $\mathbb{Q}(A_c)>0$. By the Dambis--Dubins--Schwarz theorem \cite[Theorem 18.4]{Kal2002}, there exists a standard Brownian motion $B=(B_t)_{t \geqslant 0}$ defined on an extension $\big(\bar{\Omega},\bar{\mathscr{F}},\bar{\mathbb{Q}}\big)$ of $(\Omega,\mathscr{F},\mathbb{Q})$, such that the scaled Brownian motion $B'_t \eqdefl \sqrt{c} B_t$, $t \geqslant 0$, satisfies
\begin{equation*}
S_t = S_0+B'_{c^{-1}\langle S \rangle_t} \quad \textrm{for all $t \in [0,T]$ $\bar{\mathbb{Q}}$-a.s.}\, .
\end{equation*}
Then, in view of \eqref{rho-times}, it follows that the sequence
\begin{equation*}
\xi_n \eqdefl (B'_{n^{-\gamma}} - B'_{(n+1)^{-\gamma}})^+ = (\sqrt{c}B_{n^{-\gamma}} - \sqrt{c}B_{(n+1)^{-\gamma}})^+,\quad n \in \N\, ,
\end{equation*}
satisfies 
\begin{equation}\label{Z-equiv}
Z_n = \xi_n \quad \textrm{$\bar{\mathbb{Q}}$-a.s.\ on $A_c$ for any $n \in \N$}\, . 
\end{equation}
Applying Lemma \ref{lem:BM-conditions} to the random variables $\xi_1,\xi_2,\ldots$ with $\sigma = \sqrt{c}$ and then using the equality \eqref{Z-equiv}, we deduce that, for any $\alpha \in (0,1)$,
\begin{equation}\label{Z-prop}
\sum\limits_{n=1}^\infty e^{-\alpha \sum\limits_{k=1}^n Z_k}  < \infty
\end{equation}
$\bar{\mathbb{Q}}$-a.s.\ on $A_c$.
Since the random variables $Z_1,Z_2,\ldots$ are defined on the original space $(\Omega,\mathscr{F})$, the condition \eqref{Z-prop} also holds $\mathbb{Q}$-a.s.\ on $A_c$, and thus $\prob$-a.s.\ on $A_c$ as well (due to the relation $\mathbb{Q} \sim \prob$).

We define now the process $\phi$ by
\begin{equation}\label{K-def}
\phi_t \eqdefl   \sum_{n=1}^\infty H_n\mathbf{1}_{\{ Z_n >0 \} \cap A_c} \mathbf{1}_{(\rho_{n+1},\rho_{n}]}(t), \quad t \in [0,T]\, ,
\end{equation}
where
\begin{equation*}
H_n \eqdefl \prod_{k=1}^n \frac{1}{1+\frac{2}{3} Z_k} , \quad n \in \N \, .
\end{equation*}
(Note that in \eqref{K-def}, at most one of the summands is non-zero for fixed $t$, which dispels any concerns about convergence of the random sum.) Since \eqref{Z-prop} holds $\prob$-a.s.\ on $A_c$, Lemma \ref{lem:verification} with $\alpha = \frac{1}{2}$ ensures that $\sum_{n=1}^\infty H_n < \infty$ $\prob$-a.s.\ on $A_c$, which in turn implies that the process $\phi$ is $\prob$-a.s.\ c\`agl\`ad and of finite variation  with $\phi_0=0$. Thus, by \cite[Theorems 1.2.3 and 1.2.13]{Str2011}, the stochastic integral $\int_0^t \phi_u \ud S_u$ exists as a pathwise Riemann--Stieltjes integral for any $t \in [0,T]$ and is given by
\begin{equation*}
\int_0^t \phi_u \ud S_u = \sum_{n=1}^\infty H_n\mathbf{1}_{\{ Z_n >0 \} \cap A_c} (S_{t \wedge \rho_n}-S_{t \wedge \rho_{n+1}})\, .
\end{equation*}

Let $n \in \N$ and $t \in (0,T]$. Then we have $\prob$-a.s.\ on $A_c \cap \{ \rho_{n+1} <t \leqslant\rho_n\}$,
\begin{equation*}
\begin{split}
\int_0^t \phi_u \ud S_u & = \sum_{k=n+1}^\infty H_k\mathbf{1}_{\{ Z_k >0 \}} (S_{\rho_k}-S_{\rho_{k+1}}) + H_n\mathbf{1}_{\{ Z_n >0 \}} (S_{t}-S_{\rho_{n+1}}) \\
& = \sum_{k=n+1}^\infty H_k Z_k - H_n\mathbf{1}_{\{ Z_n >0 \}}S_{\rho_{n+1}} + \phi_t S_t\, . 
\end{split}
\end{equation*}
Invoking again the fact that \eqref{Z-prop} holds $\prob$-a.s.\ on $A_c$, Lemma \ref{lem:verification} with $\alpha = \frac{1}{2}$ implies that
\begin{equation*}
\sum_{k=n+1}^\infty H_k Z_k > H_n \geqslant H_n\mathbf{1}_{\{ Z_n >0 \}}S_{\rho_{n+1}} \quad \textrm{$\prob$-a.s.\ on $A_c$}\, ,
\end{equation*}
where second inequality follows since $A_c \subset \{\sup_{t \in [0,T]} S_t \leqslant 1\} $. Note additionally that
\begin{equation*}
\int_0^t \phi_u \ud S_u = \sum_{k=1}^\infty H_k Z_k \geqslant \sum_{k=2}^\infty H_k Z_k > H_1 \geqslant 0 = \phi_t S_t \quad \textrm{$\prob$-a.s.\ on $A_c\cap \{ \rho_1 < t\}$}\, .
\end{equation*}
Therefore,
\begin{equation*}
\int_0^t \phi_u \ud S_u > \phi S_t \quad \textrm{$\prob$-a.s.\ on $A_c\cap \{ \rho < t\}$}\, ,
\end{equation*}
and since the process $\big(\int_0^t \phi_u \ud S_u\big)_{t \in [0,T]}$ is continuous and $(\phi_t S_t)_{t \in [0,T]}$ is c\`agl\`ad, we find that
\begin{equation*}
\prob\bigg( \int_0^t \phi_u \ud S_u > \phi_t S_t \textrm{ for all $t \in (\rho,T]$}\bigg) \geqslant \prob(A_c)>0\, ,
\end{equation*}
so we have established \eqref{eq:prob}. It remains to observe that
\begin{equation*}
\int_0^t \phi_u \ud S_u = 0=\phi_t S_t\, ,
\end{equation*}
$\prob$-a.s.\ on $\Omega \setminus A_c$ and when $t \leqslant \rho$, so also \eqref{eq:as} follows.
\end{proof}

\section*{Acknowledgements}

The research of J.\ Lukkarinen has been partially supported by the Academy of Finland via the Centre of Excellence in Analysis and Dynamics Research (project 271983) and from an Academy Project (project 258302). M. S. Pakkanen acknowledges partial support 
from CREATES (DNRF78), funded by the Danish National Research Foundation, 
from the Aarhus University Research Foundation (project ``Stochastic and Econometric Analysis of Commodity Markets"), and
from the Academy of Finland (project 258042).

\end{document}